\documentclass[smallextended]{svjour3}       
\smartqed  
\usepackage{graphicx}
\usepackage{amsmath,amsfonts,amssymb}
\usepackage{epstopdf}
\usepackage{comment}

\newcommand{\Hn}{\mathcal H^{(n)}}
\newcommand{\K}{\mathbf K^{(n)}_\mathfrak{T}}
\newcommand{\E}{\mathbf E^{(n)}_\mathfrak{T}}
\newcommand{\KBF}{^\alpha\mathbf{KBF}^{(n)}_\mathfrak{T}}
\newcommand{\KPF}{^\gamma\mathbf{KPF}^{(n)}_\mathfrak{T}}
\newcommand{\KBPF}{^\beta\mathbf{KBPF}^{(n)}_\mathfrak{T}}
\newcommand{\KD}{^p\mathbf{KD}^{(n)}_\mathfrak{T}}
\newcommand{\KAD}{^{p,\lambda}\mathbf{KAD}^{(n)}_\mathfrak{T}}
\newcommand{\parg}[1]{\left\{#1\right\}}  
\newcommand{\ket}[1]{\vert {#1} \rangle}  
\newcommand{\nl}{\par\noindent}   
\newcommand{\Prob}{{\tt p}}    
\newcommand{\Not}{{\tt NOT}}

\newcommand{\QXor}{{\tt XOR}}

\newcommand{\C}{\mathbb{C}}
\newcommand{\Id}{{\rm I}}

\begin{document}
	
\title{Quantum Approach to Epistemic Semantics
}
\titlerunning{Quantum approach to epistemic semantics}        

\author{Giuseppe Sergioli \and
	Roberto Leporini
}

\authorrunning{G.Sergioli, R.Leporini} 

\institute{Giuseppe Sergioli \at
	Universit\`a di Cagliari\\
	via Is Mirrionis 1, I-09123 Cagliari, Italy.\\
	\email{giuseppe.sergioli@gmail.com}
	\and
	Roberto Leporini \at
	Dipartimento di Ingegneria gestionale, dell'informazione e della produzione\\
	Universit\`a di Bergamo\\
	viale Marconi 5, I-24044 Dalmine (BG), Italy.\\
	\email{roberto.leporini@unibg.it}
}

\date{Received: date / Accepted: date}

\maketitle

\begin{abstract}
	Quantum information has suggested new forms of quantum logic, called \emph{quantum computational logics},
	where \emph{meanings} of sentences are represented by pieces of quantum information (generally, density  operators of some Hilbert spaces),
	which can be stored and transmitted by means of quantum particles.
	This approach can be applied to a semantic characterization of \emph{epistemic logical operations},
	which may occur in sentences like ``At time $\mathfrak t'$ Bob knows that at time $\mathfrak t$ Alice knows that the spin-value is up''.
	Each epistemic agent (say, \emph{Alice}, \emph{Bob},...) has a characteristic \emph{truth-perspective},
	corresponding to a particular orthonormal basis of the Hilbert space $\C^2$.
	From a physical point of view, a truth-perspective can be associated to an apparatus that allows one to measure a given observable.
	An important feature that characterizes the knowledge of any agent is the amount of information that is accessible to him/her
	(technically, a special set of density operators, which also represents the \emph{internal memory} of the agent in question).
	One can prove that interesting epistemic operations are special examples of quantum channels, which generally are not unitary.
	The \emph{act of knowing} may involve some intrinsic irreversibility due to possible measurement-procedures or
	to a loss of information about the environment.
	We also illustrate some relativistic-like effects that arise in the behavior of epistemic agents.
	
\end{abstract}

\section{Introduction}

Quantum information has recently suggested new problems and new investigations in logic.
An interesting example concerns an abstract theory of \emph{epistemic logical operations}
(like ``to know'') developed in a Hilbert space environment \cite{BDCGLS,DBGS}. Most standard approaches to epistemic logics have proposed a quite
strong characterization of epistemic operations. In such
frameworks knowledge is generally modelled  as a \emph{potential}
rather than an \emph{actual} knowledge. Accordingly, a sentence
like ``Alice knows that the spin-value in the $x$-direction is
up'' turns out to have the meaning  ``Alice \emph{could} know''
rather than ``Alice \emph{actually} knows''. A consequence of
these theories is the unrealistic phenomenon of \emph{logical
	omniscience} according to which whenever Alice knows a given
sentence, then Alice knows \emph{all} the logical consequences
thereof. Hence, for instance, knowing the axioms of Euclidean
geometry  should imply knowing all the theorems of the theory.

A more realistic logical theory of knowledge can be developed  in
the framework of \emph{quantum computational logics}
\cite{DGL05}, which are new forms of quantum logic suggested by the
theory of \emph{quantum logical gates} in quantum computation.
The basic ideas that underly the semantic characterization of
these logics can be sketched as follows. Pieces of quantum
information (mathematically represented by density operators of
convenient Hilbert spaces) are regarded as possible \emph{meanings}
for the sentences of a given formal language. At the
same time, the basic logical connectives of this language are
interpreted as particular quantum logical gates: unitary quantum
operations that transform density operators in a reversible way.
Accordingly, any sentence can be imagined as a synthetic logical
description of a quantum circuit.

Is it possible and interesting to describe \emph{logical epistemic
	operations} (say, ``At time $\mathfrak t$ Alice knows'') as a
special kind of quantum operation? We give a positive answer to this question. The intuitive idea is the
following: whenever $\rho$ represents a piece of quantum
information (for instance the qubit $\ket{1}$), the operation
$\mathbf K_{\mathfrak a_\mathfrak t}$ (say, ``At time $\mathfrak
t$ Alice knows'') transforms $\rho$ into another piece of quantum
information $\mathbf K_{\mathfrak a_\mathfrak t}\rho$, which lives
in the same space of $\rho$ and asserts that ``At time $\mathfrak
t$ Alice knows $\rho$''. Of course, generally, $\mathbf
K_{\mathfrak a_\mathfrak t}\rho$ and $\rho$ will be different
density operators. On this basis, one can study the behavior of
nested epistemic operations, like ``At time $\mathfrak t'$ Bob
knows that at time $\mathfrak t$ Alice knows $\rho$'' ($\mathbf
K_{\mathfrak b_{\mathfrak t'}} \mathbf K_{\mathfrak a_\mathfrak
	t}\rho$).

Like any quantum information, also epistemic pieces of information
(say,  $\mathbf K_{\mathfrak a_\mathfrak t}\rho$) may be \emph{true} or \emph{false} or \emph{indeterminate},
where truth-values are defined in terms of a natural notion of quantum probability.
Conventionally, one can assume that the two elements $\ket{1}$ and $\ket{0}$ of the canonical orthonormal basis
of the Hilbert space $\C^2$ represent in this framework the truth-values \emph{Truth} and \emph{Falsity}.
Accordingly (by application of the Born-rule), the probability of being true for a generic qubit
$\ket{\psi}= a_0\ket{0} + a_1\ket{1}$ will be the number ${\Prob}(\ket{\psi})= |a_1|^2$.
The definition of Born-like probabilities can be canonically extended to all density operators,
living in any product-space $\underbrace{\C^2\otimes\ldots\otimes\C^2}_{n-times}$.

The choice of an orthonormal basis for the space $\C^2$ is, obviously, a matter of convention.
One can consider infinitely many bases that are determined by the application of a unitary
operator $\mathfrak T$ to the elements of the canonical basis.
From an intuitive point of view, we can think that the operator
$\mathfrak T$ gives rise to a change of \emph{truth-perspective}.
While the canonical truth-perspective is identified with the pair of bits $\ket{1}$  and $\ket{0}$,
in the $\mathfrak T$-truth perspective \emph{Truth} and \emph{Falsity} are identified with
the two qubits $\mathfrak T \ket{1}$ and $\mathfrak T \ket{0}$, respectively.
In this framework, it is not strange to guess that
different epistemic agents may have different truth-perspectives,
corresponding to different ideas of \emph{Truth} and \emph{Falsity}.

From a physical point of view, each truth-perspective can be naturally regarded as associated to
a physical apparatus that allows one to measure a given observable.
As an example, consider a source  emitting a pair of photons correlated in polarization in
such a way  that both photons have the same polarization
(say, the horizontal polarization $\ket H$ or the vertical polarization $\ket V$).
Suppose that the two photons are in the entangled state $$\ket\psi=a\ket {H,H}+b\ket {V,V}$$
(with $|a|^2+|b|^2=1$).
The orthonormal basis $B=\{\ket H, \ket V\}$ represents here a particular truth-perspective.
Let us refer to an observer equipped with a polarizer that detects $45^o$ and $135^o$ polarized photons.
The same state $\ket\psi$ will be represented by the observer from a different truth-perspective,
corresponding to the orthonormal basis $B^\prime=\{\ket{45},\ket{135}\}$.
In fact, by applying the transformation
$$\ket H =\frac{1}{\sqrt2}(\ket{45} -\ket{135}); \;\ket V =\frac{1}{\sqrt2}(\ket{45}+\ket{135}),$$
the observer will describe the state $\ket \psi$ as
$$\ket \psi=\frac{a+b}{2}(\ket{45,45}+\ket{135,135})-\frac{a-b}{2}(\ket{45,135}+\ket{135,45}).$$
Notice that a truth-perspective change gives rise to a different
description of one and the same physical state $\ket \psi$.

An important feature that characterizes the knowledge of any agent
is represented by the amount of information that is accessible to
him/her. Technically, the \emph{epistemic domain} of an agent can
be identified with a special set of density operators. From an
intuitive point of view, this set can be regarded as the set of
pieces of information that  our agent is able to \emph{understand} and to \emph{memorize}. And the limits of
epistemic domains can be used to avoid the unpleasant phenomenon of
logical omniscience.  One can prove that epistemic operations are
not generally unitary. From an intuitive point of view, the
\emph{act of knowing} seems to involve some intrinsic
irreversibility due to possible measurement-procedures or to a
loss of information about the environment.

In Section 4 we will see how quantum channels are deeply connected with epistemic operations:
special examples of irreversible quantum operations that can be used to reach the information
stored by density operators.
Finally 
we will also illustrate some relativistic-like effects that arise in the
behavior of epistemic agents.

An abstract study of epistemic operations in a Hilbert-space
environment may have a double interest:
\begin{enumerate}
	\item from a logical point of view such analysis shows how
	``thinking in a quantum-theoretic way'' can contribute to
	overcome some crucial difficulties of standard epistemic
	logics.
	\item From a physical point of view this analysis stimulates
	further investigations about possible correlations between
	the irreversibility of quantum operations and the kind of
	``jumps'' that seem to characterize acts of knowledge,
	both in the case of human and of artificial intelligence.
\end{enumerate}

\section{Truth-perspectives and quantum logical gates}
The general mathematical environment is the $n$-fold tensor product of the Hilbert space $\C^2$:
$$\Hn:=\underbrace{\C^2\otimes\ldots\otimes\C^2}_{n-times},$$
where all pieces of quantum information live.
The elements $\ket{1} = (0,1)$ and $\ket{0} = (1,0)$ of the canonical orthonormal basis $B^{(1)}$
of $\C^2$ represent, in this framework, the two classical bits,
which can be also regarded as the canonical truth-values \emph{Truth} and \emph{Falsity}, respectively.
The canonical basis of $\Hn$ is the set
$$B^{(n)} =\parg{\ket{x_1}\otimes \ldots \otimes \ket{x_n}: \ket{x_1},\ldots, \ket{x_n} \in B^{(1)}}.$$
As usual, we will briefly write $\ket{x_1, \ldots, x_n}$ instead of $\ket{x_1} \otimes \ldots \otimes \ket{x_n}$.
By definition, a \emph{quregister} is a unit vector of $\Hn$.
Quregisters thus correspond to pure states, namely to maximal pieces of information about the quantum
systems that are supposed to store a given amount of quantum information.
We shall also make reference to \emph{mixtures} of quregisters,
represented by density operators $\rho$ of $\Hn$.
Of course, quregisters  correspond to special cases of density operators.
We will denote by $\mathfrak D(\Hn)$ the set of all density operators of $\Hn$,
while $\mathfrak D= \bigcup_n\parg{\mathfrak D(\Hn)}$
will represent the set of all possible pieces of quantum information.

As observed in the Introduction, from an intuitive point of view,
a basis-change in $\C^2$ can be regarded as a change of
\emph{truth-perspective}. While in the classical case, the
truth-values \emph{Truth} and \emph{Falsity} are identified with
the two classical bits $\ket{1}$ and $\ket{0}$, assuming a
different basis corresponds to a different idea of \emph{Truth}
and \emph{Falsity}. Since any basis-change in $\C^2$ is determined
by a unitary operator, we can identify a \emph{truth-perspective}
with a unitary operator $\frak T$ of $\C^2$. We will write:
$$ \ket{1_{\frak T}}= \frak T \ket{1};\; \ket{0_{\frak T}}=\frak T\ket{0},$$ and
we will assume that $ \ket{1_{\frak T}}$ and $\ket{0_{\frak T}}$ represent, respectively,
the truth-values \emph{Truth}  and \emph{Falsity} of the truth-perspective $\frak T$.
The \emph{canonical truth-perspective} is, of course,
determined by the identity operator ${\tt I}$ of $\C^2$.
We will indicate by $B^{(1)}_\frak T$ the orthonormal basis determined by $\frak T$;
while $B^{(1)}_{\tt I}$ will represent the canonical basis.

Any unitary operator $\frak T$ of $\mathcal H^{(1)}$ can be naturally extended to a unitary operator
$\frak T^{(n)}$ of $\Hn$ (for any $n \geq 1$):
$$\frak T^{(n)}\ket{x_1,\ldots,x_n}= \frak T\ket{x_1}\otimes \ldots \otimes \frak T\ket{x_n}.$$

Accordingly, any choice of a unitary operator $\frak T$ of $\mathcal H^{(1)}$
determines an orthonormal basis $B^{(n)}_\frak T$ for $\Hn$ such that:
$$B^{(n)}_\frak T = \parg{\frak T^{(n)}\ket{x_1,\ldots,x_n}:\ket{x_1,\ldots,x_n} \in B_{\tt I}^{(n)}}.$$
Instead of $\frak T^{(n)}\ket{x_1,\ldots,x_n}$ we will also write
$\ket{x_{1_{\frak T}},\ldots,x_{n_{\frak T}}}$.

The elements of $B^{(1)}_\frak T$ will be called the $\frak T$-{\it bits} of $\mathcal H^{(1)}$;
while the elements of $B^{(n)}_\frak T$ will represent the $\frak T$-{\it registers} of $\Hn$.

On this ground the notions of \emph{truth}, \emph{falsity} and \emph{probability} with respect to any
truth-perspective $\mathfrak T $ can be defined in a natural way.

\begin{definition} \emph{($\frak T$-true and $\frak T$-false registers)}\label{def:true}\nl
	\begin{itemize}
		\item $\ket{x_{1_{\frak T}},\ldots,x_{n_{\frak T}}}$ is a \emph{$\frak T$-true register} iff $\ket{x_{n_{\frak T}}} =\ket{1_{\frak T}};$
		\item $\ket{x_{1_{\frak T}},\ldots,x_{n_{\frak T}}}$ is a \emph{$\frak T$-false register} iff $\ket{x_{n_{\frak T}}} =\ket{0_{\frak T}}.$
	\end{itemize}
\end{definition}

In other words, the \emph{$\frak T$-truth-value} of a $\frak T$-register
(which corresponds to a sequence of $\frak T$-bits) is determined by its last element.
\footnote{As we will see, the application of a classical reversible gate to a register $\ket{x_1,\ldots,x_n}$
	transforms the (canonical) bit $\ket{x_n}$ into the target-bit $\ket{x_n^\prime}$, which behaves as the final truth-value.
	This justifies our choice in Definition \ref{def:true}.}

\begin{definition} \emph{($\mathfrak T$-truth and $\mathfrak T$-falsity)}\nl
	\begin{itemize}
		\item The $\mathfrak T$-\emph{truth} of $\Hn$ is the projection operator $^\mathfrak TP_1^{(n)}$
		that projects over the closed subspace spanned by the set of all $\mathfrak T$- true registers;
		\item the $\mathfrak T$-\emph{falsity} of $\Hn$ is the projection operator $^\mathfrak TP_0^{(n)}$
		that projects over the closed subspace spanned by the set of
		all $\mathfrak T$- false registers.
	\end{itemize}
\end{definition}

In this way, truth and falsity are dealt with as mathematical
representatives of possible physical properties. Accordingly, by
applying the Born-rule, one can naturally define the
probability-value of any density operator with respect to the
truth-perspective $\mathfrak T$.

\begin{definition} \emph{(}$\mathfrak T$-\emph{Probability)}\nl
	For any $\rho \in \mathfrak D(\Hn)$,
	$${\Prob}_\mathfrak T(\rho) := {\tt Tr}(^\mathfrak TP_1^{(n)} \rho),$$
	where ${\tt Tr}$ is the trace-functional.
\end{definition}
We interpret ${\Prob}_\mathfrak T(\rho)$ as the probability that the information $\rho$ satisfies the $\mathfrak T$-Truth.

In the particular case of qubits, we will obviously obtain:
$${\Prob}_\mathfrak T(a_0\ket{0_\mathfrak T} + a_1\ket{1_\mathfrak T}) =|a_1|^2.$$
Two truth-perspectives $\mathfrak {T}_1$ and $\mathfrak {T}_2$ are called \emph{probabilistically equivalent} iff for any density operator $\rho$, $\tt p_{\mathfrak{T}_1}(\rho)=\tt p_{\mathfrak{T}_2}(\rho)$.

For any choice of a truth-perspective  $\mathfrak T$, the set $\mathfrak D$ of all density operators can be pre-ordered by a
relation that is defined in  terms of the probability-function ${\Prob}_\mathfrak T$.

\begin{definition}\emph{(Preorder)} \label{de:preordine}\nl
	$\rho \preceq_\mathfrak T \sigma$ iff ${\Prob}_\mathfrak T(\rho) \le {\Prob}_\mathfrak T(\sigma)$.
\end{definition}

This preorder relation plays an important role in the semantics of
quantum computational logics. For, the {\em logical
	consequence}-relation between sentences is defined in terms of
$\preceq_\mathfrak T$ (see Section 4).

As is well known, quantum information is processed by \emph{quantum logical gates} (briefly, \emph{gates}):
unitary operators that transform quregisters into quregisters in a reversible way.

In this article we will consider some well known quantum gates \cite{DCGLS14}: the \emph{negation} $\tt NOT^{(n)}$, the \emph{Toffoli gate} $\tt T^{(n,m,p)}$, the \emph{controlled-not gate} 
$\tt XOR^{(n,m)}$, the \emph{Hadamard-gate} $\sqrt{\tt I}^{(n)}$ and the \emph{square root of the negation} $\sqrt{\tt NOT}^{(n)}$, that play a special role
both from the computational and from the logical point of view.

All gates can be naturally transposed from the canonical truth-perspective to any truth-perspective $\mathfrak T$.
Let $G^{(n)}$ be any gate defined with respect to the canonical truth-perspective.
The \emph{twin-gate} $G^{(n)}_{\mathfrak T}$, defined with respect to the truth-perspective $\mathfrak T$, is determined as follows:
$$G^{(n)}_{\mathfrak T}:=\mathfrak T^{(n)} G^{(n)} \mathfrak T^{{(n)\dagger}},$$
where $\mathfrak T^\dagger$ is the adjoint of $\mathfrak T$.

All $\mathfrak T$-gates, defined on $\Hn$, can be canonically extended to
the set of all density operators of $\Hn$.
Let $G_\mathfrak T$ be any gate defined on $\Hn$.
The corresponding \emph{unitary quantum operation} $^\mathfrak DG_\mathfrak T$ is defined as follows
for any $\rho \in \mathfrak D(\Hn)$:
$$^{\mathfrak D}G_\mathfrak T\rho=G_\mathfrak T\,\rho\, G_\mathfrak T^\dagger.$$

It is interesting to consider a convenient notion of \emph{distance} between truth-perspectives.
As is well known, different definitions of distance between vectors can be found in the literature.
For our aims it is convenient to adopt the Fubini-Study definition of distance between two qubits.
\begin{definition} (\emph{The Fubini-Study distance})\label{de:fubini}\nl
	Let $\ket{\psi}$ and $\ket{\varphi}$ be two qubits.
	$$d(\ket{\psi}, \ket{\varphi})=\frac{2}{\pi}\arccos|\langle \psi|\varphi\rangle|.$$
\end{definition}

This notion of distance satisfies the following conditions:
\begin{enumerate}
	\item $d(\ket{\psi}, \ket{\varphi})$ is a metric distance;
	\item $\ket{\psi} \perp  \ket{\varphi} \Rightarrow d(\ket{\psi}, \ket{\varphi})=1 $;
	\item $d(\ket{1}, \ket{1_{Bell}}) = \frac {1}{2}$, where $\ket{1}$  is the canonical truth,
	while $\ket{1_{Bell}}=\sqrt{\tt I}^{(1)}\ket{1}=\left(\frac{1}{\sqrt{2}},-\frac{1}{\sqrt{2}}\right)$
	represents the \emph{Bell-truth} (which corresponds to a maximal uncertainty with respect to the canonical truth).
\end{enumerate}

On this ground, one can naturally define the \emph{epistemic distance} between two truth-perspectives.

\begin{definition} (\emph{Epistemic distance})\label{de:epdist}\nl
	Let $\mathfrak T_1$ and $\mathfrak T_2$  be two truth-perspectives.
	$$d^{Ep}(\mathfrak T_1, \mathfrak T_2)=d(\ket{1_{\mathfrak T_1}}, \ket{1_{\mathfrak T_2}}).$$
\end{definition}

In other words, the epistemic distance between the truth-perspectives $\mathfrak T_1$ and $\mathfrak T_2$
is identified with the distance between the two qubits that represent the truth-value \emph{Truth} in
$\mathfrak T_1$ and in $\mathfrak T_2$, respectively.

\section{Epistemic operations and epistemic structures}
We will now introduce the concepts of \emph{(logical) epistemic operation} and of \emph{epistemic structure}.

\begin{definition} \emph{((Logical) epistemic operation and strong epistemic operation)} \label{de:know}\nl
	A \emph{(logical) epistemic operation} of the space $\Hn$ with respect to the truth-perspective $\frak T$ is a map
	$$\E: \mathcal B(\Hn) \mapsto \mathcal B(\Hn),$$ where
	$\mathcal B(\Hn)$ is the set of all bounded operators of $\Hn$.
	The following conditions are required:
	\begin{enumerate}
		\item $\E$ is associated with an \emph{epistemic domain}
		$EpD(\E)$, which is a subset of $\frak D(\Hn)$;
		\item for any $\rho \in \mathfrak D(\Hn)$, $\E\rho \in \mathfrak D(\mathcal H^{(n)})$;
		\item $\forall \rho \in \frak D(\Hn): \rho\notin EpD(\E)\;\Rightarrow
		\;\E\rho = \overline{\rho_0}$
		(where $\overline{\rho_0}$ is a fixed density operator of $\mathfrak D(\Hn)$).
	\end{enumerate}
	An epistemic operation $\E$ is called \emph{strong} iff
	$\E \rho \preceq_\mathfrak T \rho$,
	for any $\rho \in EpD(\mathbf E_{\frak T})$
	(where $\preceq_\mathfrak T$ is the preorder relation defined by Def.\ref{de:preordine}).
	
\end{definition}
As expected, an intuitive interpretation of $\E\rho$ is the following:
``the piece of information $\rho$ is known/believed''.
The strong epistemic operation described by $\E$ is limited by a given epistemic domain
(which is intended to  represent  the information accessible to a given agent,
or also his/her \emph{memory}, relatively to the space $\mathcal H^{(n)}$).
Whenever a piece of information $\rho$ does not belong to the epistemic domain of
$\E$, then $\E\rho$ collapses into a fixed element $\overline{\rho_0}$
(which may be identified, for instance, with the maximally uncertain information $\frac{1}{2^n}{\tt I}^{(n)}$ or
with the $\mathfrak T$-Falsity $\frac{1}{2^{n-1}}\,^\mathfrak TP_0^{(n)}$).
At the same time, whenever $\rho$  belongs to the epistemic domain of $\E$,
it seems reasonable to assume that the probability-values of $\rho$ and
$\E\rho$ are correlated:
the probability of the quantum information asserting that ``$\rho$ is known/believed'' should
always be less than or equal to the probability of $\rho$.
Hence, in particular, we have:
$$\Prob_{\frak T}(\E\rho) =1 \;\Rightarrow\;\Prob_{\frak T}(\rho)=1.$$
But generally, not the other way around! In other words, pieces of
quantum information that are certainly known are certainly true
(with respect to the truth-perspective in question).

A strong epistemic operation $\E$  is called \emph{non-trivial}
iff for at least one density operator $\rho \in EpD(\E)$, 
${\Prob}_{\frak T}(\E\rho) < {\Prob}_{\frak T}(\rho)$.
Notice that strong epistemic operations do not generally preserve
pure states \cite{BDCGLS,BDGLS2,BDGLS3}.

One can prove that non-trivial epistemic operations cannot be
represented by unitary quantum operations, being generally
irreversible \cite{BDCGLS}. Their behavior is, in a sense, similar
to the behavior of measurement-operations.

At the same time, some interesting epistemic operations can be represented by
the more general notion of \emph{quantum channel} defined below.
\footnote{Quantum channels represent particular cases of \emph{quantum operations}.
	The concept of quantum operation is a quite general notion that permits us to represent at the same time
	\emph{symmetry transformations of quantum states}, \emph{effects} and \emph{measurements}.
	In particular, it has been shown that for \emph{open systems},
	interacting with an environment, the Schr\"odinger-equation should be generalized to a superoperator-equation,
	describing how an initial pure state evolves into a mixed  state, transformation that has the form of a quantum operation.
	See, for instance,\cite{CDP},\cite{Hong}.}

\begin{definition} \emph{(Quantum channel)}\label{quantochann}
	\footnote{This definition is based on the so called \emph{Kraus
			first representation theorem}. See \cite{K}.} \nl A \emph{quantum
		channel} on $\Hn$ is a linear map $\mathcal E^{(n)}$ from
	$\mathcal B(\Hn)$ to $\mathcal B(\Hn)$ such that for some set $I$
	of indices there exists a set $\parg{E_{i}}_{i \in I}$ of elements
	of $\mathcal B(\Hn)$ satisfying the following conditions:
	\begin{enumerate}
		\item $\sum_i E_i^\dagger E_i = {\tt I}^{(n)}$;
		\item $\forall A \in  \mathcal B(\Hn): \mathcal E^{(n)}(A) = \sum_i E_iAE_i^\dagger$.
	\end{enumerate}
\end{definition}

A set $\parg{E_{i}}_{i \in I}$ such that $\sum_i E_i^\dagger E_i =
{\tt I}^{(n)}$ is usually called a \emph{system of Kraus
	operators}. One can prove that quantum channels are
trace-preserving, and hence transform density operators into
density operators.

Of course, unitary quantum operations $^\mathfrak D G^{(n)}$ are special cases of quantum channels,
for which $\parg{E_{i}}_{i \in I}= \parg{G^{(n)}}$.
As expected, quantum channels can be defined with respect to
any truth-perspective $\mathfrak T$.

We define now some possible properties of epistemic operations
that have a significant logical interest. One is dealing with
strong conditions that are generally violated in the real use of
our intuitive notion of  knowledge.

\begin{definition}\emph{(Introspection, consistency and monotonicity)}
	\label{de:intr}\nl
	A strong epistemic operation $\E$ of $\Hn$
	(with respect to the truth-perspective $\frak T$) is called
	\begin{enumerate}
		\item \emph{positively introspective} iff for any $\rho\in\mathfrak D (\Hn):$
		$${\Prob}_{\frak T}(\E\rho) \le {\Prob}_{\frak T}
		(\E\E\rho).$$ In other words, whenever we know/believe we know/believe
		that we know/believe.
		\item \emph{Negatively introspective} iff for any
		$\rho\in\mathfrak D(\Hn):$
		$${\Prob}_{\frak T}(\,^{\frak D}\Not_\frak T\E\rho)
		\le {\Prob}_{\frak T}(\E\,^{\frak D}\Not_\frak T\E\rho).$$
		In other words, whenever we do not know/believe we know/believe that we do not
		know/believe.
		
		\item \emph{Probabilistically consistent} iff for any
		$\rho\in\mathfrak D(\Hn):$
		$${\Prob}_{\frak T}(\E\rho) \le {\Prob}_{\frak T}(\,^{\frak D}\Not_\frak T\E
		\,^{\frak D}\Not_\frak T\rho).$$  In other words,
		whenever we know/believe an information we do not know/believe its negation.
		\item \emph{monotonic} iff for any $\rho, \sigma \in\mathfrak
		D(\Hn):$
		$${\Prob}_{\frak T}(\rho) \le {\Prob}_{\frak
			T}(\sigma)\,\,\Rightarrow \,\,
		{\Prob}_{\frak T}(\E\rho) \le {\Prob}_{\frak T}(\E\sigma).$$
	\end{enumerate}
\end{definition}

Using the concepts defined above, we can  finally introduce the
notion of \emph{epistemic quantum computational structure}. From
an intuitive point of view, an epistemic quantum computational
structure can be described as a system consisting of a set of
\emph{epistemic agents} evolving in time (where time is dealt with
as a finite sequence of instants $(\mathfrak t_1, \ldots,
\mathfrak t_n)$). Each agent $\mathfrak a$ is characterized by a
truth-perspective $\mathfrak T_\mathfrak a$, which (for the sake
of simplicity) is supposed to be constant in time. At any time
$\mathfrak t$ and for any Hilbert space $\Hn$, each agent is
associated to a strong epistemic operation $\mathbf E^{(n)}_{\mathfrak
	{T_{a}},\mathfrak {a_t}}$, whose epistemic domain represents the
amount of information that our agent is able to \emph{understand}
and to \emph{memorize} at that particular time (relatively to the
space $\Hn$). When $\rho$ belongs to the epistemic domain of
$\mathbf E^{(n)}_{\frak {T_{a}},\mathfrak {a_t}}$, then the number
${\Prob}_{\mathfrak T_{\mathfrak a}}(\mathbf E^{(n)}_{\mathfrak
	{T_{a}},\mathfrak {a_t}}\rho)$ represents the probability that
agent $\mathfrak a$ at time $\mathfrak t$ knows/believes the quantum
information $\rho$.

\begin{definition}\emph{(Epistemic quantum computational structure)}\nl
	An \emph{epistemic quantum computational structure} is a system
	$$\mathcal S=(T,\,Ag,\,TrPersp,\,\mathbf{Inf},\,\mathbf U,\,\mathbf B,\,\mathbf K),$$
	where:
	\begin{enumerate}
		\item $T$ is a time-sequence $(\mathfrak t_1,\ldots,\mathfrak t_n)$;
		\item $Ag$ is a set of epistemic agents where each agent $\mathfrak a$ is represented as a function of the time $\frak t$ in $T$.
		We will write $\mathfrak {a_t}$ instead of $\mathfrak{a(t)}$;
		\item $TrPersp$ is a map that assigns to any agent $\mathfrak a$ a truth-perspective $\frak T_\mathfrak a$
		(the truth-perspective of $\mathfrak a$);
		\item $\mathbf{Inf}$ is a map that assigns to any $\mathfrak{a_t}$ and to any
		$n \geq 1$ a map, called (logical) \emph{information operation}
		$$\mathbf {Inf}^{(n)}_{\frak {T_{a}},\mathfrak {a_t}}:\mathcal B(\Hn) \mapsto \mathcal B(\mathcal H^{(n)}),$$
		which is an epistemic operation with respect to the truth-perspective $\frak {T_{a}}$ (the truth-perspective of agent $\frak a$).
		\item $\mathbf U$ is a map that assigns to any $\mathfrak{a_t}$ and to any
		$n \geq 1$ a map, called (logical) \emph{understandig operation}
		$$\mathbf U^{(n)}_{\frak {T_{a}},\mathfrak {a_t}}:\mathcal B(\Hn) \mapsto \mathcal B(\mathcal H^{(n)}),$$
		which is an epistemic operation with respect to the truth-perspective $\frak {T_{a}}$.
		\item $\mathbf B$ is a map that assigns to any $\mathfrak{a_t}$ and to any
		$n \geq 1$ a map, called (logical) \emph{belief operation}
		$$\mathbf B^{(n)}_{\frak {T_{a}},\mathfrak {a_t}}:\mathcal B(\Hn) \mapsto \mathcal B(\mathcal H^{(n)}),$$
		which is a strong epistemic operation with respect to the truth-perspective $\frak {T_{a}}$.
		\item $\mathbf K$ is a map that assigns to any $\mathfrak{a_t}$ and to any
		$n \geq 1$ a map, called (logical) \emph{knowledge operation}
		$$\mathbf K^{(n)}_{\frak {T_{a}},\mathfrak {a_t}}:\mathcal B(\Hn) \mapsto \mathcal B(\mathcal H^{(n)}),$$
		which is a strong epistemic operation with respect to the truth-perspective $\frak {T_{a}}$.
		The following conditions are required:
		\subitem{(i)} $\mathbf K^{(n)}_{\frak {T_{a}},\mathfrak {a_t}} \rho \preceq_\mathfrak T \mathbf B^{(n)}_{\frak {T_{a}},\mathfrak {a_t}}\rho$, for any $\rho \in EpD(\mathbf K^{(n)}_{\frak {T_{a}},\mathfrak {a_t}});$
		\subitem{(ii)} $\mathbf B^{(n)}_{\frak {T_{a}},\mathfrak {a_t}}\rho \preceq_\mathfrak T \mathbf U^{(n)}_{\frak {T_{a}},\mathfrak {a_t}}\rho$, for any $\rho \in EpD(\mathbf B^{(n)}_{\frak {T_{a}},\mathfrak {a_t}})$
		(where $\preceq_\mathfrak T$ is the preorder relation defined by Def.\ref{de:preordine});
		\subitem{(iii)} $EpD(\mathbf K^{(n)}_{\frak {T_{a}},\mathfrak {a_t}})\subseteq EpD(\mathbf B^{(n)}_{\frak {T_{a}},\mathfrak {a_t}})\subseteq EpD(\mathbf U^{(n)}_{\frak {T_{a}},\mathfrak {a_t}})\subseteq EpD(\mathbf{Inf}^{(n)}_{\frak {T_{a}},\mathfrak {a_t}})$.
	\end{enumerate}
\end{definition}
From an epistemic point of view, an agent can be certain to have an uncertain information and vice versa.

\section{Quantum noise channels as epistemic operations}

We will now illustrate some examples of strong epistemic operations that may be interesting
from a physical point of view.
One is dealing with special cases of \emph{quantum noise channels}, which can be, generally,
obtained from some unitary operators, tracing out  the ancillary qubits that describe the environment.

Let $\alpha, \beta, \gamma$ be complex numbers such that $|\alpha|^2 + |\beta|^2 + |\gamma|^2 \le 1$.
Consider the following system of Kraus operators:
$$\begin{array}{l}
E_0=\sqrt{1-|\alpha|^2-|\beta|^2-|\gamma|^2}\,{\tt I}\\
E_1=|\alpha|\sigma_x\\
E_2=|\beta|\sigma_y\\
E_3=|\gamma|\sigma_z
\end{array}$$
(where $\sigma_x$, $\sigma_y$, $\sigma_z$ are the three Pauli matrices).\\
Define $^{\alpha,\beta,\gamma}\mathcal E^{(1)}$ as follows for any $\rho \in \mathfrak D(\C^2)$:
$$\, ^{\alpha,\beta,\gamma}\mathcal E^{(1)} \rho=\sum_{i=0}^3 E_i\, \rho\, E_i^{\dagger}.$$
We have:
$$^{\alpha,\beta,\gamma}\mathcal E^{(1)} \rho=(1-|\alpha|^2-|\beta|^2-|\gamma|^2)\rho+|\alpha|^2\sigma_x\rho\sigma_x+
|\beta|^2\sigma_y\rho\sigma_y+|\gamma|^2\sigma_z\rho\sigma_z$$

One can prove that for any choice of $\alpha$, $\beta$, $\gamma$
(such that $|\alpha|^2 + |\beta|^2 + |\gamma|^2 \le 1$),
the map $^{\alpha,\beta,\gamma}\mathcal E^{(1)}$ is a quantum channel of the space $\C^2$.

Let us refer to  the Bloch-sphere  corresponding to $\mathfrak D(\C^2)$.
Any map $^{\alpha,\beta,\gamma}\mathcal E^{(1)}$ induces the following vector-transformation
(the sphere is deformed into an ellipsoid centered at the origin):
$$\left(\begin{array}{l}
x\\
y\\
z
\end{array}\right)
\mapsto
\left(\begin{array}{l}
(1-2 |\beta|^2-2 |\gamma|^2)\; x\\
(1-2 |\alpha|^2-2 |\gamma|^2)\; y\\
(1-2|\alpha|^2-2 |\beta|^2)\; z
\end{array}\right)$$

For particular choices of  $\alpha$, $\beta$ and $\gamma$, one obtains some special cases of quantum channels.
\begin{itemize}
	\item For $\alpha=\beta=\gamma=0$, one obtains the identity operator.
	\item For $\beta=\gamma=0$, one obtains the \emph{bit-flip
		channel} $^\alpha\mathcal{BF}^{(1)}$ that flips the two
	canonical bits (represented as the projection operators
	$^{\tt I}P_0^{(1)}$ and  $^{\tt I}P_1^{(1)}$) with
	probability $|\alpha|^2$:
	$$\,^{\tt I}P_0^{(1)} \mapsto (1-|\alpha|^2) \,^{\tt I}P_0^{(1)}+(|\alpha|^2) \,^{\tt I}P_1^{(1)};$$
	$$\,^{\tt I}P_1^{(1)} \mapsto (1-|\alpha|^2)\,^{\tt I}P_1^{(1)}+(|\alpha|^2) \,^{\tt I}P_0^{(1)}.$$
	The sphere is mapped into an ellipsoid with $x$ as symmetry-axis 
\end{itemize}


\begin{itemize}
	\item For $\alpha=\gamma=0$, one obtains the \emph{bit-phase-flip channel} $^\beta\mathcal{BPF}^{(1)}$ that
	flips both bits and phase with probability $|\beta|^2$. The sphere is mapped into an ellipsoid with $y$ as symmetry-axis.
	\item For $\alpha=\beta=0$, one obtains the \emph{phase-flip channel} $^\gamma\mathcal{PF}^{(1)}$
	that flips the phase with probability $|\gamma|^2$.
	The sphere is mapped into an ellipsoid with $z$ as symmetry-axis.
	\item For $|\alpha|^2=|\beta|^2=|\gamma|^2=\frac{p}{4}$, one obtains the \emph{depolarizing channel} $^p\mathcal{D}^{(1)}$.
	If $p=1$, the polarization along any direction is equal to $0$.
	The sphere is contracted by a factor $1-p$ and the center of the sphere is a fixed point.
\end{itemize}

Another interesting example is the \emph{generalized amplitude damping} $\mathcal {AD}^{(1)}$
induced by the following system of Kraus-operators:
$$\begin{array}{ll}
E_0=\sqrt{\lambda}\left(\begin{array}{cc}1 & 0\\ 0 & \sqrt{1-p} \end{array}\right)
&E_2=\sqrt{1-\lambda}\left(\begin{array}{cc}\sqrt{1-p} & 0 \\ 0& 1\end{array}\right)\\
E_1=\sqrt{\lambda}\left(\begin{array}{cc}0 & \sqrt{p} \\ 0& 0\end{array}\right)
&E_3=\sqrt{1-\lambda}\left(\begin{array}{cc}0 & 0 \\ \sqrt{p}& 0\end{array}\right)
\end{array}$$
where $\lambda,p\in[0,1]$.\nl
This channel determines the following transformation on the Bloch-sphere (see Fig.~1):
$$\left(\begin{array}{l}
x\\
y\\
z
\end{array}\right)
\mapsto
\left(\begin{array}{l}
\sqrt{1-p}\; x\\
\sqrt{1-p}\; y\\
(1-p)\; z + p (2\lambda-1)
\end{array}\right)
$$
\begin{figure}[h]
	\centering
	\includegraphics[scale=0.4]{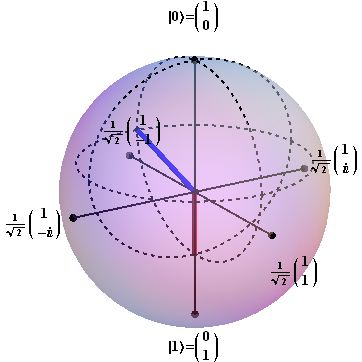}
	\caption{The amplitude damping channel}\label{AD}
\end{figure}

The channels we have considered above have been defined with respect to the canonical truth-perspective ${\tt I}$.
However, as expected, they can be naturally transposed to any truth-perspective $\mathfrak T$.
Given $\mathcal E^{(1)}$ such that $\mathcal E^{(1)}\rho = \sum_{i=0}^3 E_i\, \rho\, E_i^{\dagger}$,
the twin-channel $ \mathcal E^{(1)}_\mathfrak T$ of $\mathcal E^{(1)}$ can be defined as follows:
$$\mathcal E^{(1)}_\mathfrak T\rho:=\sum_i\mathfrak T E_i\mathfrak T^\dagger\,\rho\,\mathfrak T E_i^{\dagger}\mathfrak T^\dagger.$$

So far we have only considered quantum channels of the space $\C^2$.
At the same time, any operation $\mathcal E^{(1)}_\mathfrak T$
(defined on $\C^2$) can be canonically extended to an operation $\mathcal E^{(n)}_\mathfrak T$ defined on the space $\Hn$ (for any $n > 1$).
Consider a density operator $\rho$ of $\Hn$ and let $Red^n(\rho)$ represent the reduced state of the $n$-th subsystem of $\rho$.
We have: ${\Prob}_\mathfrak T(\rho) = {\tt Tr}(^\mathfrak T P_1^{(1)} Red^n(\rho))$.
In other words, the $\mathfrak T$-probability of $\rho$ only depends on the $\mathfrak T$-probability of the $n$-th subsystem of $\rho$.
On this basis, it is reasonable to define $\mathcal E^{(n)}_\mathfrak T$  as follows:
$$\mathcal E^{(n)}_\mathfrak T\,= \, {\tt I}^{(n-1)} \otimes\mathcal E^{(1)}_\mathfrak T.$$

Notice that, generally, a  quantum channel $\mathcal E^{(n)}$ does
not represent a strong epistemic operation. We may have, for a $\rho$
(which is supposed to belong to the epistemic domain):
$${\Prob}_{\tt I}(^\alpha\mathcal {BF}^{(1)}\rho) \nleq {\Prob}_{\tt I}(\rho),  $$
against the definition of strong epistemic operation.

At the same time, by convenient choices of the epistemic domains,
our quantum channels can be transformed into strong epistemic operations.

\begin{definition} \emph{(A bit-flip epistemic operation $\KBF$)} \label{de:bit}\nl
	Let $\alpha \neq 0$. Define $\KBF$ as follows:
	\begin{enumerate}
		\item $EpD(\KBF) \subseteq
		D = \{\rho\in \mathfrak D(\Hn)\, |\,\Prob_{\mathfrak T}(\rho)\geq \frac{1}{2}\}$.\nl
		In other words an agent (whose strong epistemic operation is $\KBF$)
		only understands pieces of information that are not ``too far from the truth''.
		\item $\rho \in EpD(\KBF) \;\Rightarrow \;
		\KBF\rho = \,\,^\alpha\mathcal {BF}^{(n)}_\mathfrak T \rho$.
	\end{enumerate}
\end{definition}

\begin{theorem} \label{th:bit}\cite{DCGLS14}\nl
	\begin{enumerate}
		\item[(i)]Any $\KBF$ is a strong epistemic operation. In particular,
		$\KBF$ is a non-trivial epistemic operation if there
		exists at least one $\rho\in EpD(\KBF)$ such that
		$\Prob_{\mathfrak T}(\rho)>\frac{1}{2}$.
		\item[(ii)] the set $D$ is the maximal set such that the corresponding
		$\KBF$ is a strong epistemic operation.
		\item[(iii)] Let $|\alpha|^2 \le \frac{1}{2}$ and let
		$EpD(\KBF) = D$. The following closure property holds: for
		any $\rho\in D$, $\KBF \rho\in D$.
	\end{enumerate}
\end{theorem}

In a similar way one can define strong epistemic operations that correspond to the phase-flip channel, the bit-phase-flip channel,
the depolarizing channel and the generalized amplitude-damping channel.

\begin{definition} \emph{(A phase-flip epistemic operation $\KPF$) }\nl
	Let $\gamma \neq 0$. Define $\KPF$ as follows:
	\begin{enumerate}
		\item $EpD(\KPF) \subseteq \mathfrak D(\Hn)$.
		\item $\rho \in EpD(\KPF) \;\Rightarrow \;
		\KPF \rho = \;^\gamma\mathcal{PF}^{(n)}_\mathfrak T \rho$.
	\end{enumerate}
\end{definition}

\begin{theorem}\label{th:kpf}\nl
	\nl ${\Prob}_\mathfrak T(\KPF\rho)
	= \Prob_\mathfrak T(\rho)$, for any $\rho \in EpD(\KPF)$.\nl
	Hence, $\KPF$ is a trivial epistemic operation.
	
\end{theorem}
\begin{proof}\nl
	(i)-(ii) Suppose $\rho\in EpD(\KPF)\subseteq \mathfrak D(\Hn)$.
	Let us consider ${\frak T}^\dagger Red^n(\rho){\frak
		T}=\frac{1}{2}(\Id+x \sigma_x+y \sigma_y+z \sigma_z)$. We have,
	$\Prob_{\frak T}(\KPF\rho)={\tt Tr}(\,^{\frak
		T}P_1^{(n)}\\ \KPF\rho) ={\tt Tr}(\,^\Id P_1^{(1)} \sum_{i}
	E_i\mathfrak T^\dagger\, Red^n(\rho)\, \mathfrak T E_i^{\dagger})
	=\frac{1-z}{2}=\Prob_{\frak T}(\rho)$.
\end{proof}

\begin{definition} \emph{(A bit-phase-flip epistemic operation $\KBPF$) }\nl
	Let $\beta \neq 0$.
	Define $\KBPF$ as follows:
	\begin{enumerate}
		\item $EpD(\KBPF) \subseteq
		D=\{\rho\in \mathfrak D(\Hn)\, |\,\Prob_{\mathfrak T}(\rho)\geq \frac{1}{2}\}$.
		\item $\rho \in EpD(\KBPF)\;\Rightarrow\;
		\KBPF\rho = \;^\beta\mathcal{BPF}^{(n)}_\mathfrak T \rho$.
	\end{enumerate}
\end{definition}

\begin{theorem}\nl
	\begin{enumerate}
		\item[(i)] Any $\KBPF$ is a strong epistemic operation. In
		particular, $\KBPF$ is a non-trivial epistemic operation
		if there exists at least one $\rho\in EpD(\KBPF)$ such
		that $\Prob_{\mathfrak T}(\rho)>\frac{1}{2}$.
		\item[(ii)] the set $D$ is the maximal set such that the corresponding
		$\KBPF$ is a strong epistemic operation.
		\item[(iii)] Let $|\beta|^2 \le \frac{1}{2}$ and let
		$EpD(\KBPF) = D$. The following closure property holds:
		for any $\rho\in D$, $\KBF \rho\in D$.
	\end{enumerate}
\end{theorem}
\begin{proof}\nl
	\nl Similar to the proof of Theorem \ref{th:bit}.
\end{proof}

\begin{definition} \emph{(A depolarizing epistemic operation $\KD$)}\nl
	Let $p \neq 0$.
	Define $\KD$ as follows:
	\begin{enumerate}
		\item $EpD(\KD) \subseteq
		D =\{\rho\in \mathfrak D(\Hn)\, |\,\Prob_{\mathfrak T}(\rho)\geq \frac{1}{2}\}$.
		\item $\rho \in EpD(\KD)\;\Rightarrow \;
		\KD\rho = \;^p\mathcal {D}^{(n)}_\mathfrak T \rho$.
	\end{enumerate}
\end{definition}
Notice that for any truth-perspectives
$\mathfrak T$, $^p\mathcal {D}^{(n)}_\mathfrak T=\,^p\mathcal {D}^{(n)}_\Id$.

\begin{theorem}\label{th:depol}\nl
	\begin{enumerate}
		\item[(i)] Any $\KD$ is a strong epistemic operation. In particular,
		$\KD$ is a non-trivial epistemic operation if there exists
		at least one $\rho\in EpD(\KD)$ such that
		$\Prob_{\mathfrak T}(\rho)>\frac{1}{2}$.
		\item[(ii)] the set $D$ is the maximal set such that the corresponding
		$\KD$ is a strong epistemic operation.
		\item[(iii)] Let $EpD(\KD)=D$. Then, for any $\rho\in
		EpD(\KD)$, we have: $\KD \rho \in EpD(\KD).$
	\end{enumerate}
\end{theorem}
\begin{proof}\nl
	(i)-(ii) Suppose $\rho\in EpD(\KD)\subseteq D$. Let us consider
	${\frak T}^\dagger Red^n(\rho){\frak T}=\frac{1}{2}(\Id+x \sigma_x+y \sigma_y+z \sigma_z)$.
	We have, $\Prob_{\frak T}(\KD\rho)=\frac{1-(1-p)z}{2}$. Hence,
	$\KD \rho \preceq_\mathfrak T \rho \Leftrightarrow (1-p)z\ge z \Leftrightarrow z\in[-1,0] \Leftrightarrow
	\Prob_{\mathfrak T}(\rho)\geq \frac{1}{2}$. Thus, $\KD$ is an epistemic operation.\\
	(iii) $\Prob_{\frak T}(\KD\rho)=\frac{1-(1-p)z}{2}\ge \frac{1}{2}$, since $z\in[-1,0]$.\\
\end{proof}

\begin{definition} \emph{(A generalized amplitude-damping epistemic operation $\KAD$})\nl
	Let $p, \lambda \in[0,1]$. Define $\KAD$ as follows:
	\begin{enumerate}
		\item $EpD(\KAD) \subseteq
		AD = \{\rho\in \mathfrak D(\Hn)\,|\, \Prob_{\mathfrak T}(\rho)\geq 1 - \lambda\}$.\nl
		\item $\rho \in EpD(\KAD) \;\Rightarrow\;
		\KAD\rho =\,^{p,\lambda}\mathcal {AD}^{(n)}_\mathfrak T\rho$.
	\end{enumerate}
\end{definition}

\begin{theorem}\nl
	\begin{enumerate}
		\item[(i)] Any $\KAD$ is a strong epistemic operation. In particular,
		$\KAD$ is a non-trivial epistemic operation if there
		exists at least one $\rho\in EpD(\KAD)$ such that
		$\Prob_{\mathfrak T}(\rho)>1-\lambda$.
		\item[(ii)] the set $AD$ is the maximal set such that the corresponding
		$\KAD$ is a strong epistemic operation.
		\item[(iii)] Let $EpD(\KAD)=AD$. Then, for any  $\rho\in
		EpD(\KAD)$, we have:  $\KAD \rho\in EpD(\KAD)$.
	\end{enumerate}
\end{theorem}

\begin{proof}\nl
	(i)-(ii) Suppose $\rho\in EpD(\KAD)\subseteq AD$. Let us consider
	${\frak T}^\dagger Red^n(\rho){\frak T}=\frac{1}{2}(\Id+x \sigma_x+y \sigma_y+z \sigma_z)$.
	We have, $\Prob_{\frak T}(\KAD\rho)=\frac{1-(1-p)z-p(2\lambda-1)}{2}$.
	Hence,
	$\KAD \rho \preceq_\mathfrak T \rho \Leftrightarrow (1-p)z+p(2\lambda-1)\ge z \Leftrightarrow 2\lambda-1\ge z \Leftrightarrow \Prob_{\mathfrak T}(\rho)\geq 1-\lambda$. Thus, $\KAD$ is a strong epistemic operation.\\
	(iii) $\Prob_{\frak T}(\KAD\rho)=\frac{1-(1-p)z-p(2\lambda-1)}{2}\ge 1-\lambda$, since $2\lambda-1\ge z$.\\
\end{proof}

The following theorem sums up some interesting properties of the strong epistemic operations defined above.

\begin{theorem}\nl
	\begin{itemize}
		\item[(i)] All strong epistemic operations
		$\KBF$,
		$\KPF$,
		$\KBPF$,
		$\KD$ are probabilistically consistent.
		\item[(ii)] All strong epistemic operations
		$\KAD$ with
		$\lambda\ge\frac{1}{2}$ are probabilistically consistent.
		\item[(iii)] All strong epistemic operations
		$\KPF$ are monotonic,
		positively and negatively introspective.
		\item[(iv)] All strong epistemic operations $\KBF$
		(with $|\alpha|^2\le \frac{1}{2}$),
		$\KBPF$
		(with $|\beta|^2 \le \frac{1}{2}$), $\KD$,
		whose epistemic domain is a subset of
		$D=\{\rho\in \mathfrak D(\Hn)\, |\, \Prob_{\mathfrak T}(\rho)\geq \frac{1}{2}\}$ are monotonic.
		\item[(v)] All strong epistemic operations
		$\KAD$ whose epistemic domain is a subset of
		$AD = \{\rho\in \mathfrak D(\mathcal H^{(n)})\, |\, \Prob_{\mathfrak T}(\rho)\geq 1-\lambda\}$ are monotonic.
	\end{itemize}
\end{theorem}
\begin{proof}\nl
	\nl Let  $\K \in \parg{\KBF,
		\KPF,
		\KBPF,
		\KD}$ and suppose that $\K$ is induced by the following
	Kraus operators:
	$$E_0=\sqrt{1-|\alpha|^2-|\beta|^2-|\gamma|^2}\,I;\,\,
	E_1=|\alpha|\sigma_x;\,\,E_2=|\beta|\sigma_y;\,\,
	E_3=|\gamma|\sigma_z.$$  Assume  that the fixed element
	$\overline{\rho_0}$ of our epistemic operations is $^{\frak T}
	P^{(n)}_0$. \nl For any $\rho, \,\rho' \in  \mathfrak D(\Hn)$, we
	can write:\nl ${\frak T}^\dagger Red^n(\rho){\frak
		T}=\frac{1}{2}(\Id+x \sigma_x+y \sigma_y+z \sigma_z);$\nl ${\frak
		T}^\dagger Red^n(\rho'){\frak T}=\frac{1}{2}(\Id+x' \sigma_x+y'
	\sigma_y+z' \sigma_z).$

	\begin{enumerate}
		\item [(i)] Let $\rho \in \mathfrak D(\Hn)$. We have: \nl
		$\rho\not\in EpD(\K)\,\,\Rightarrow \,\, \Prob_{\frak
			T}(\K\rho)=0;$\nl $^{\frak D}\Not_\frak T\rho\not\in EpD(\K)
		\,\, \Rightarrow \Prob_{\frak T}(\,^{\frak D}\Not_\frak T \K
		\,^{\frak D}\Not_\frak T\rho)=1.$ \nl Otherwise, we obtain:\nl
		$\Prob_{\frak T}(\,^{\frak D}\Not_\frak T \K \,^{\frak
			D}\Not_\frak T\rho)=$\nl ${\tt Tr}(\,^\Id P_1^{(1)}
		\Not^{(1)}\sum_{i} (E_i \Not^{(1)}\mathfrak T^\dagger\,
		Red^n(\rho)\, \mathfrak T \Not^{(1)} E_i^{\dagger})
		\Not^{(1)}) =\frac{1-(1-2 |\alpha|^2-2 |\beta|^2) z}{2}
		=\Prob_{\frak T}(\K\rho)$.\nl Consequently:  $ \Prob_\mathfrak
		T(\K\rho) \le \Prob_{\frak T}(\,^{\frak D}\Not_\frak T \K
		\,^{\frak D}\Not_\frak T\rho)$.
		\item [(ii)] 
		Similarly.
		
		\item [(iii)] By Theorem \ref{th:kpf}.
		\item [(iv)] Let $\rho,\rho'\in EpD(\K)$. Suppose that
		$|\alpha|^2+|\beta|^2=\frac{1}{2}$. Then we have:
		$\Prob_{\frak T}(\K\rho)=\Prob_{\frak
			T}(\K\rho')=\frac{1}{2}$. Otherwise we obtain: \nl
		$\Prob_{\frak T}(\K\rho)\le\Prob_{\frak T}(\K\rho')
		\Leftrightarrow \frac{1-(1-2 |\alpha|^2-2 |\beta|^2)
			z}{2}\le \frac{1-(1-2 |\alpha|^2-2 |\beta|^2) z'}{2}
		\Leftrightarrow \frac{1-z}{2}\le \frac{1-z'}{2}
		\Leftrightarrow \Prob_{\frak T}(\rho)\le\Prob_{\frak
			T}(\rho')$.
		\item [(v)] 
		Similarly.
		
	\end{enumerate}
\end{proof}

Notice that in  the general case the monotonicity-property can be violated by strong epistemic operations. In fact, the following situation is possible:
\begin{itemize}
	\item $\rho \preceq_\mathfrak T \sigma$;
	\item $\rho \in EpD(\K)$; $\sigma\notin EpD(\K)$;
	\item $\K\rho \npreceq_{\mathfrak T} \K\sigma$.
\end{itemize}

Generally, the strong epistemic operations  $\KBF$,  $\KBPF$,  $\KD$,
$\KAD$ are neither positively introspective nor negatively
introspective. Suppose, for example, that $EpD(\K)=\{\K \,
^\mathfrak T P^{(n)}_1,\, ^\mathfrak T P^{(n)}_1 \}$, where $$\K\in\{\KBF,\,
\KBPF,\,\KD,\,\KAD\}.$$ By definition of $\K$ we have for any
$\rho \in EpD(\K)$: $\Prob_\mathfrak T(\K\rho) < \Prob_\mathfrak T
(\rho)$. Hence, $\Prob_{\frak T}(\K \, ^\mathfrak T
P^{(n)}_1)>\Prob_{\frak T}(\K\K \, ^\mathfrak T P^{(n)}_1 )$.
Thus, $\K$ is not positively introspective. Suppose then that
$EpD(\K)=\{\, ^{\frak D}\Not_\frak T \K \, ^\mathfrak T
P^{(n)}_0\}$. We have: ${\Prob}_{\frak T}(\,^{\frak D}\Not_\frak
T\K \, ^\mathfrak T P^{(n)}_0)>{\Prob}_{\frak T}(\K\,^{\frak
	D}\Not_\frak T\K \, ^\mathfrak T P^{(n)}_0)$, where $\K \, ^\mathfrak T P^{(n)}_0=\, ^\mathfrak T P^{(n)}_0$.
Thus, $\K$ is not
negatively introspective.


Truth-perspectives are, in a sense, similar to different frames of reference in relativity.
Accordingly, one could try and apply a ``relativistic'' way of thinking in order to describe how
a given agent can ``see'' the logical behavior of another agent.

As expected, the  logical behavior of any agent turns out to
depend, in this framework,  on two factors:
\begin{itemize}
	\item his/her idea of \emph{Truth} and \emph{Falsity};
	\item his/her choice of the gates that correspond to the basic logical operations.
\end{itemize}
Both these factors are, of course, determined by the agent's truth-perspective $\mathfrak T$.

As an example let us refer to two agents \emph{Alice} and \emph{Bob},
whose truth-perspectives are  $\mathfrak T_{Alice}$ and $\mathfrak T_{Bob}$, respectively.
Let $\parg{\ket{1_{Alice}}, \ket{0_{Alice}}}$ and  $\{\ket{1_{Bob}},$ $\ket{0_{Bob}}\}$
represent the systems of truth-values of our two agents.
Furthermore, for any canonical gate $^\mathfrak D G^{(n)}$
(defined with respect to the canonical truth-perspective $\tt I$),
let $^\mathfrak D G^{(n)}_{Alice}$ and $^\mathfrak D G^{(n)}_{Bob}$ represent the
corresponding \emph{twin-gates} for \emph{Alice} and for \emph{Bob}, respectively.

According to the rule assumed in Section 2, we have:
$$^\mathfrak D G^{(n)}_{Alice} \,=\, ^\mathfrak D(\mathfrak T_{Alice}^{(n)}G^{(n)}\mathfrak T^{(n)\dagger}_{Alice}).$$

In a similar way in the case of \emph{Bob}.

We will adopt the following conventional terminology.
\begin{itemize}
	\item When $\ket{1_{Bob}}= a_0\ket{0_{Alice}} + a_1\ket{1_{Alice}} $,
	we will say that \emph{Alice sees that Bob's Truth is $a_0\ket{0_{Alice}} + a_1\ket{1_{Alice}}$}.
	In a similar way, for \emph{Bob}'s \emph{Falsity}.
	\item When $^\mathfrak D G^{(n)}_{Alice}  \,=\, ^\mathfrak D(\mathfrak T_{Alice}^{(n)}G^{(n)}\mathfrak T^{(n)\dagger}_{Alice})$ and
	$^\mathfrak D G^{(n)}_{Bob} \,=\,
	^\mathfrak D(\mathfrak T_{Bob}^{(n)}G^{(n)}\mathfrak T^{(n)\dagger}_{Bob})\,
	= \, ^\mathfrak D G^{(n)}_{1_{Alice}}$ (where $^\mathfrak
	D G^{(n)}$ and $^\mathfrak D G^{(n)}_{1_{Alice}}$  are
	canonical gates),
	we will say that \emph{Alice sees Bob using the gate $^\mathfrak D G^{(n)}_{1_{Alice}}$
		in place of her gate $^\mathfrak D G^{(n)}_{Alice}$}.
	\item When $^\mathfrak D G^{(n)}_{Alice}  \,=\, ^\mathfrak D G^{(n)}_{Bob}$ we will say that
	\emph{Alice} and \emph{Bob} \emph{see and use the same gate},
	which represents (in their truth-perspective) the canonical gate $^\mathfrak D G^{(n)}$.
\end{itemize}

On this basis, one can conclude that, generally, \emph{Alice}
\emph{sees} a kind of  ``deformation'' in \emph{Bob}'s logical
behavior. 

In \cite{DCGLS14}  some example have been discussed.

As another insightful example, consider the behavior of the controlled-not gate described by the following theorem.

\begin{theorem}
	Let $\mathfrak{T_a}$ be the truth-perspective of an agent $\mathfrak a$.
	\nl If $\,\,^\mathfrak D \tt{\QXor}^{(1,1)}_{\mathfrak{T_a}}\,=
	\,^\mathfrak D \tt{\QXor}^{(1,1)}$, then $\mathfrak{T_a}\,=\,e^{i\theta}\tt I.$
\end{theorem}
\begin{proof}
	Suppose $\,\,^\mathfrak D \tt{\QXor}^{(1,1)}_{\mathfrak{T_a}}\,=
	\,^\mathfrak D \tt{\QXor}^{(1,1)}$. One can easily show that:\nl
	$(\mathfrak{T_a} \otimes \mathfrak{T_a})\tt{\QXor}^{(1,1)} =
	\tt{\QXor}^{(1,1)}(\mathfrak{T_a} \otimes \mathfrak{T_a})$.
	Consequently, by standard algebraic calculations we obtain:
	$\mathfrak{T_a}\,=\,e^{i\theta}\tt I.$
	
\end{proof}
As a consequence, one immediately obtains that two agents
$\mathfrak a$ and $\mathfrak b$ can see and use the same
$\tt{\QXor}$-gate only if  their truth-perspectives
$\mathfrak{T_a}$ and $\mathfrak{T_b}$ are probabilistically
equivalent.

Moreover, a relativistic way of thinking can also be applied to strong epistemic operations.
Alice sees Bob using a phase-flip channel instead of a bit-flip channel as strong epistemic operation.
Similarly, some other agent sees Bob's strong epistemic operation acts as a bit-phase-flip channel.

From a logical point of view, examples of epistemic situations  as the one we have
here investigated, can be formally reconstructed in the framework
of a \emph{quantum computational semantics} (see \cite{DBGS}). Let
us briefly recall  the basic ideas of this approach. We consider
an epistemic quantum computational language $\mathcal L^{Ep}$
consisting of:
\begin{itemize}
	\item atomic sentences;
	\item logical connectives corresponding to the following
	gates: negation, Toffoli, controlled-not, Hadamard and square
	root of negation;
	\item names for epistemic agents (say, Alice, Bob, ...);
	\item logical epistemic operators, corresponding to (generally
	irreversible) epistemic operations.
	
\end{itemize}

This language can express sentences like ``Alice knows that Bob
does not know that the spin-value in the $x$-direction is up''.
The semantics for $\mathcal L^{Ep}$ provides a convenient notion
of \emph{model}, whose role is assigning  informational meanings to
all sentences. Technically, a model of $\mathcal L^{Ep}$ is
defined as a map ${\tt Mod}$ that associates to any
truth-perspective $\mathfrak T$ and to any sentence $\alpha$ a
density operator $\rho = {\tt Mod_{\mathfrak T}}(\alpha)$, living in a Hilbert
space $\mathcal H^\alpha$, whose dimension depends on the
linguistic complexity of $\alpha$.

On this basis, one can give a natural
definition for the concepts of \emph{truth} and  of \emph{logical
	consequence} in terms of the notions of $\mathfrak
T$-probability (${\tt p}_\mathfrak T$) and of $\mathfrak
T$-preorder ($\preceq_\mathfrak T$):
\begin{itemize}
	\item a sentence $\alpha$ is \emph{true} with respect to a
	model ${\tt Mod}$ and to a truth-perspective $\mathfrak T$
	(abbreviated as $\vDash_{{\tt Mod}, \mathfrak T} \alpha$)
	iff ${\tt p}_\mathfrak T({\tt Mod_{\mathfrak T}}(\alpha)) = 1$;
	\item a sentence $\beta$ is a \emph{logical consequence} of a
	sentence $\alpha$ (abbreviated as $\alpha \vDash \beta$)
	iff for any  model ${\tt Mod}$ and any truth-perspective
	$\mathfrak T$, $${\tt
		Mod}_\mathfrak T(\alpha) \preceq_\mathfrak T  {\tt Mod}_\mathfrak T(\beta).$$

\end{itemize}

While truth is obviously dependent on the choice of a truth
perspective ${\mathfrak T}$, one can prove that the notion of
logical consequence represents an \emph{absolute} relation that
is invariant with respect to truth-perspective changes.

\begin{theorem}\cite{DBGS} \nl
	$\alpha \vDash \beta$ iff for any model ${\tt Mod}$, ${\tt Mod}_{\tt I}(\alpha) \preceq_{\tt I} {\tt Mod}_{\tt I}(\beta)$, where $\tt I$ is the
	canonical truth-perspective.
\end{theorem}

On this basis one can conclude that:

\begin{itemize}
	\item \emph{Alice} and \emph{Bob} may have different ideas about
	the logical connectives, about truth, falsity and probability.
	\item In spite of these differences, the
	\emph{reasoning-rules} (which are determined by the
	logical consequence relation) are the same for
	\emph{Alice} and for \emph{Bob}. Apparently, assigning the
	same interpretation to the logical connectives is not a
	necessary condition in order to use the same
	reasoning-rules.
\end{itemize}

\begin{acknowledgements}
	We would like to thank  Enrico Beltrametti, Maria Luisa Dalla Chiara and Roberto Giuntini, who have deeply discussed this paper with us, proposing some useful suggestions.
	Sergioli's work has been supported by the Italian Ministry of Scientific Research within the FIRB project ``Structures and dynamics of knowledge and cognition", Cagliari unit F21J12000140001 and by Regione Sardegna within the project ``Modeling the uncertainty: Quantum Theory and Imaging Processing"; Leporini's work has been supported by the Italian Ministry of Scientific Research within the PRIN project ``Automata and Formal Languages: Mathematical Aspects and Applications".
\end{acknowledgements}

\end{document}